\documentclass[twoside,a4paper,12pt]{amsart}

\usepackage[english]{babel}
\usepackage{amsmath}
\usepackage{amsthm}
\usepackage{amssymb}
\usepackage{hyperref}
\usepackage{xcolor}
\usepackage{graphicx}

\newtheorem{theorem}{Theorem}

\theoremstyle{remark}
\newtheorem{remark}{Remark}

\def\S{\mathcal S}

\newcommand{\pd}[2]{\frac{\partial {#1}}{\partial {#2}}}
\newcommand{\be}{\begin{equation}}
\newcommand{\ee}{\end{equation}}

\newcommand{\n}{\noindent}
\newcommand{\sgn}{\operatorname{sgn}\,}
\newcommand{\supp}{\operatorname{supp}}
\newcommand{\RE}{\mathbb R}
\newcommand{\CO}{\mathbb C}

\newcommand{\HH}{\mathcal H}

\newcommand{\uz}{\underline{z}}
\newcommand{\uy}{\underline{y}}
\newcommand{\ue}{\underline{e}}

\newcommand{\nn}{^{(n)}}

\newcommand{\comment}[1]{}

\title[]{Effective equation for a system of mechanical oscillators in an acoustic field}

\author[]{Claudio Cacciapuoti}
\address{Hausdorff Center for Mathematics, Institute for Applied Mathematics, Bonn Universit\"at, Endenicher Allee 60, 53115 Bonn, Germany}
\curraddr{Dipartimento di Scienza e Alta Tecnologia, Universit\`a dell'Insubria, Via Valleggio 11, 22100 Como, Italy}
\email{claudio.cacciapuoti@uninsubria.it}

\author[]{Rodolfo Figari}
\address{Dipartimento di Scienze Fisiche, Universit\`a di Napoli Federico II, Napoli, Italy}
\address{Istituto Nazionale di Fisica Nucleare, Sezione di Napoli, Via Cintia, 80126 Napoli, Italy}
\email{figari@na.infn.it}

\author[]{Andrea Posilicano}
\address{Dipartimento di Scienza e Alta Tecnologia, Universit\`a dell'Insubria, Via Valleggio 11, 22100 Como, Italy}
\email{posilicano@uninsubria.it}

\thanks{
C.C. acknowledges the support of the FIR 2013 project ``Condensed Matter in Mathematical Physics'' (code RBFR13NAET)}


\begin{document}
\maketitle

\begin{abstract}
We consider a one dimensional evolution problem modeling the dynamics of an acoustic field coupled with a set of mechanical oscillators. We analyze solutions of the system of ordinary and partial differential equations with time-dependent boundary conditions describing the evolution in the limit of a continuous distribution of oscillators. 
\end{abstract}



\section{Introduction}
The investigation of the dynamics of spherical bubbles oscillating in a surrounding compressible fluid has been the subject of many papers in Applied Acoustics. Starting from the pioneering works \cite{C.B, V.B1, V.B2} of C.A. Bjerknes and V.F.K. Bjerknes, at the beginning of last century,  the analysis of the forces (secondary Bjerknes forces) exchanged between bubbles via the self-generated acoustic field was the basis of various theoretical attempt to examine the evolution of ``bubbly liquids''.  Here, we limit ourselves to cite a few review papers  \cite{DZ95, D05, HOP01, P04} and a more recent work  \cite{L12} where a thorough list of references on various aspects of this topic can be found.

\n
The case of spherical bubbles of small linear size  in comparison with their average mutual distances was the first to be analyzed. Under this hypothesis, only the far field emitted by the bubbles is relevant to the interaction forces and the bubbles can be considered point-like. Nevertheless, at the best of our knowledge, no attempts has been done to investigate the equations governing the evolution of the acoustic field and of the bubble dynamics in the limit of zero bubble radii. 

\n
It is worth mentioning that, at the times when the Bjerknes' started their investigations on the acoustic field driven interactions between bubbles, many physicists were approaching the apparently much harder problem of producing a theory of the electromagnetic field together with its point sources (see e.g. \cite{Abr1,  Abr2, Lor1, Lor2, Lor3}).

\n
In the following, we present a one dimensional model of a coupled system of acoustic field and mechanical oscillators. Based on the results we obtained in \cite{CFP06}, for the case of a finite number of oscillators, we address the problem of deducing the asymptotic form of the velocity field evolution equation when the oscillators are continuously distributed in a finite interval.

\n 
In our opinion, the asymptotic problem has reasons of mathematical interest. In fact, the small length scale does not refer here to the typical range of some inhomogeneity of the medium  inside which the field evolves, but to the mutual distances of the sources producing and interacting with the field itself. In this respect, the problem is in no way typical of an homogenization problem for an hyperbolic system of evolution equations (see e.g. \cite{CDMZ,  DFM, GPX, PL, R})

 In fact, we investigate here a scattering process: the limit Cauchy problem  describes the evolution of the acoustic field velocity from one side to the other of a region filled with mechanical oscillators reacting to the field and becoming sources of secondary acoustic waves. At the best of our knowledge rigorous results on limit effective equations for problems of this kind are not available.

\n We describe first the physical setup our system of equations is meant to model. We consider the following arrangement: an infinite pipe  filled with a non-viscous, compressible fluid and $n$ mechanical oscillators  realized with thin walls of mass $M_j$ positioned  in the  pipe perpendicularly to  its axis.  A system of springs with 
elastic constants $K_j$ confines the walls around their equilibrium positions 
located in the points of coordinates $s_j$, with respect to a system of coordinates whose 
$x$-axis coincides with the axis of the pipe.

\n We suppose that there is no friction between the fluid and the
pipe and we analyze only one dimensional flows. 
The acoustic  field is then described by the  pressure field $p(x,t)$
and  the  velocity  field  $v(x,t)$. The  motion  of  the  mechanical
oscillators is described through the displacements $y_j(t)$ of the thin walls from their
equilibrium positions $s_j$ and the velocities
$z_j(t)=\dot y_j(t)$, $j=1,\dots,n$. 

\begin{center}
\begin{figure}[h!]
\includegraphics[scale=1]{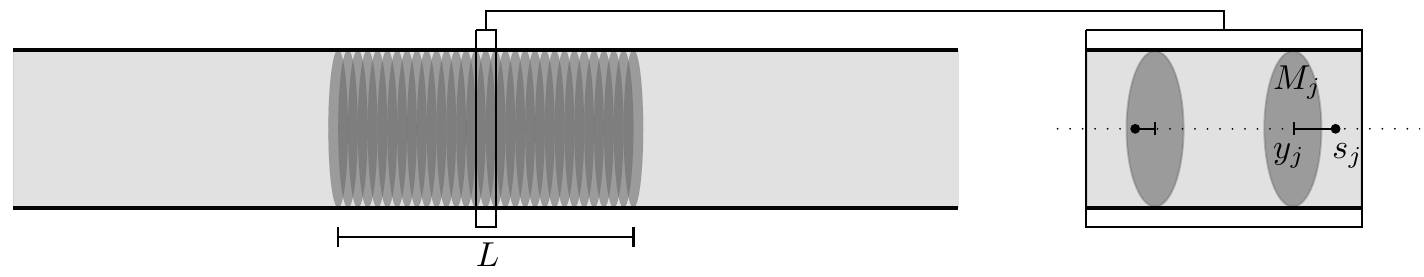}
\caption{The figure shows the cross section of the whole system (on the left) and the details of two oscillators (on the right).}
\end{figure}
\end{center}

\n The field $p(x,t)$ specifies the deviation of the pressure in the
point $x$ at time $t$ with respect 
to an equilibrium pressure $P_0$. We analyze the dynamics of the whole
system (acoustic field and mechanical oscillators) in the linearized
acoustic and elastic regimes.  Moreover  we consider  walls  of  zero
  thickness and we suppose that the contact between the fluid and the walls 
  is maintained throughout the evolution.
  
\n  Under  these
hypotheses the  dynamics is described  by the following system  of ordinary
and  partial  differential  equations  with  time  dependent  boundary
conditions
\begin{align}
&\pd{p}{t}=-a^2\rho_0\pd{v}{x}& x\in\RE\backslash\S\label{globeq1}\\
&\pd{v}{t}=-\frac{1}{\rho_0}\pd{p}{x}& x\in\RE\backslash\S \label{globeq2}\\
&M_j\frac{dz_j}{dt}=-K_jy_j-S\left(p(s_j^+)-p(s_j^-)\right)&j=1,\dots,n \label{globeq3}\\
&z_j=\frac{dy_j}{dt}=v(s_j)&j=1,\dots,n\label{globeq4}
\end{align}
\n where $\rho_0$ is  a positive constant representing the 
density  of  the  fluid, $a$  is  a  positive  constant
indicating the
 velocity of sound in the fluid, $S$ is the area of the transversal
 section of the pipe and $\S:=(s_1,\dots,s_n)$ is the set of the equilibrium 
 positions of the walls. 
 Equations \eqref{globeq3} and \eqref{globeq4} express the Newton's laws for the mechanical oscillators. The two forces acting on each wall are the spring elastic force and the force due to the difference of pressure on its opposite sides.
 
 \n Aim of this paper is to characterize the effective behavior of the system 
in the limit of a continuous  distribution of oscillators in a bounded region.
 
\n Given the initial conditions
\be\label{inco}
p(x,0);\;v(x,0);\; y_j(0);\;z_j(0)\qquad j=1,\dots,n
\ee
\n we consider the
Cauchy problem (\ref{globeq1})-(\ref{globeq4}), (\ref{inco}).

\n  As it  was shown  in \cite{CFP06}  existence and  unicity,  together with detailed 
properties of the solutions, can  be proved formulating the  problem as a unitary
evolution in the space of finite energy. In the following, we fix the notation and we recall the
main result obtained in \cite{CFP06}.

\n We indicate with a capital Greek letter the generic vector
$\Psi := (p,v,\uy,\uz)\in L^2(\RE)\oplus L^2(\RE)\oplus\CO^n\oplus\CO^n\,,$
where $L^2(\RE)$ is the space of square-integrable functions on
the real line,
\[
\uy=y_1\ue_1+\dots+y_n\ue_n\,,\quad 
\uz=z_1\ue_1+\dots+ z_n\ue_n 
\] 
\n and $\ue_1,\dots,\ue_n$ is the canonical
orthonormal base in $\CO^n$. 

\n The Hilbert space of finite energy is defined by 
\[
 \HH:=L^2(\RE)\oplus L^2(\RE)\oplus\CO^n\oplus\CO^n
\]
with the scalar product 
\[
\langle\langle\Psi_1,\Psi_2\rangle\rangle:=
\frac{1}{a^2\rho_0}\,\langle p_1,p_2\rangle+\rho_0
\langle v_1,v_2\rangle+\frac{1}{S}\sum_{j=1}^{n}\big[{K_j}\, \bar y_{1j}
y_{2j}
+{M_j}\,\bar z_{1j}z_{2j}\big]\,,
\]
where  $\langle\cdot\,,\cdot\rangle$  indicates  the  standard scalar  product  in
$L^2(\RE)$ and $^-$ denotes 
complex conjugation.
 
\n      The     square      norm     of      a      
vector     $\Psi$,
$\|\Psi\|^2=\langle\langle\Psi,\Psi\rangle\rangle$, 
defines the total energy of the system in the state $\Psi$
\[
E_{tot}=\frac{S}{2}\|\Psi\|^2=E_{ac}+E_{osc}\,,
\]
\n with
\begin{align*}
E_{ac}&=\frac{S}{2a^2\rho_0}\int_{-\infty}^\infty
|p(x)|^2dx+\frac{S\rho_0}{2}\int_{-\infty}^\infty
|v(x)|^2dx\\
E_{osc}&=\sum_{j=1}^n\bigg[\frac{K_j}{2}|y_j|^2+\frac{M_j}{2}|z_j|^2\bigg]\,,
\end{align*}
\n $E_{ac}$ is the energy stored in the acoustic field while $E_{osc}$
is the energy of the mechanical oscillators.

\n Let $\bar H^1(\RE)$ be the homogeneous Sobolev space of
locally square-integrable functions with square-integrable (distributional) 
derivative and $H^1(\RE)$ the usual Sobolev space $H^1(\RE)
:=\bar H^1(\RE)\cap
L^2(\RE)$. Similarly we define, for any open set $\mathcal O\subseteq\RE$, the Sobolev spaces 
$H^1(\mathcal O):=\{f\in L^2(\mathcal O)\,:\, \frac{df}{dx}\in L^2(\mathcal O)\}$ and 
$H^2(\mathcal O):=\{f\in H^{1}(\mathcal O)\,:\, \frac{d^2f}{dx^2}\in L^2(\mathcal O)\}$. 
\n The following theorem holds
\begin{theorem}
\label{th1}
The linear operator
\[
\hat A:D(\hat A)\subset\HH\to\HH
\]
\[
\begin{aligned}
D(\hat A)
:=\big\{&(p,v,\uy,\uz)\,\big|\, p\in L^2(\RE)\cap H^1(\RE\backslash \S),\, v\in
H^1(\RE),\\
&\  \uy\in\CO^n,\, \uz\in\CO^n,\, 
v(s_j)=z_j,\,j=1,\dots,n\big\}
\end{aligned}
\]
\[
\hat A(p,v,\uy,\uz):=\left(-a^2\rho_0\,\frac{dv}{dx},\,
-\frac{1}{\rho_0}\,\frac{dp_{reg}}{dx},\,\uz,\,
-\sum_{j=1}^{n}\left(\frac{K_j}{M_j}\,y_j+\frac{S}{M_j}\,\sigma_j\right)\,\ue_j\right)
\]
is real and skew-adjoint. Here $\sigma_j\in \CO$ is defined by
\[
\sigma_j:=p(s_j^+)-p(s_j^-)
\]
and $p_{reg}\in \bar H^1(\RE)$ 
is the regular part of $p$, i.e.  
\[
p_{reg}(x):=p(x)-\frac{1}{2}\sum_{j=1}^{n}\sigma_j\,
\sgn(x-s_j)\,.
\]
\end{theorem}

\n
\begin{remark}
Notice that the left and right limits entering in the 
definitions of the $\sigma_j$'s exist and are finite by 
\[
|p(x_1)-p(x_2)|\le (x_2-x_1)^{1/2}\,\left\|\frac{dp}{dx}
\right\|_{L^2(s_j,s_{j+1})}\,,\quad s_j<x_1\le x_2<s_{j+1}\,.
\]
\end{remark}
\n Theorem \ref{th1} implies that the Cauchy problem 
\begin{align}
&\frac{d}{dt}\Psi_t=\hat A\Psi_t\label{caop1}\\
&\Psi_{t=0}=\Psi_0\qquad\Psi_0\in\HH\label{caop2}
\end{align}
\n has a unique solution expressed via the strongly continuous
unitary group  of evolution generated  by $\hat A$,  
$\Psi_t=e^{t\hat A
 }\Psi_0$.  Moreover the unitary evolution preserves reality in such a way that one  
 can consider the flow restricted to the real Hilbert space 
\[
\HH_r:=L_r^2(\RE)\oplus L_r^2(\RE)\oplus\RE^n\oplus\RE^n\,,
\]
\n where $L^2_r(\RE)$ is the space of real-valued functions in
$L^2(\RE)$. The operator $\hat A$ is skew-adjoint so that 
$\|\Psi_t\|=\|\Psi_0\|$, 
and  the  total energy  remains
constant.   It  is  easy to check that   the unique solution of
(\ref{caop1})-(\ref{caop2}) solves in $ \HH$ 
the  Cauchy problem (\ref{globeq1})-(\ref{globeq4}),  (\ref{inco}). 

\n In fact  from Theorem \ref{th1} the equation (\ref{caop1}) is equivalent to the system
\begin{align*}
&\pd{p}{t}=-a^2\rho_0\,\pd{v}{x}\,,
\\
&\pd{v}{t}=-\frac{1}{\rho_0}\,\pd{p_{reg}}{x}\,,
\\
&\frac{dy_j}{dt}=z_j\,,\quad j=1,\dots,n
\,,\\
&\frac{dz_j}{dt}=-\left(\frac{K_j}{M_j}\,y_j+\frac{S}{M_j}\,\sigma_j\right)\,,\quad j=1,\dots,n\,.
\end{align*}
\n Notice that if $\Psi_0\in D(\hat A)$ the solution $\Psi_t$ 
belongs to $D(\hat A)$ so that condition (\ref{globeq4}) is  automatically satisfied at any time.

\n
In the following section we investigate the effective response to an incoming acoustical wave 
of a chain of mechanical oscillators in the limit of a continuous distribution of oscillator masses.

\section{Effective equation in the limit of a continuous distribution of oscillators}
 We consider a sequence of Cauchy problems indexed by the increasing number $(n)$ of mechanical oscillators
\begin{align} 
&\frac{d}{dt}\Psi_t\nn=\hat A\nn\Psi_t\nn
 \label{cpn1}
\\
&\Psi_{t=0}\nn=\Psi_0\nn\qquad\Psi_0\nn\in\HH
\label{cpn2}
\end{align}
where we denoted, as above, with $\Psi_t\nn$ the quadruple $\Psi_t\nn := (p_t\nn,v_t\nn,\uy_t\nn,\uz_t\nn)$. We investigate the limit, for increasing $n$, of the evolution equation satisfied by the velocity field alone.

\n
We will make the following assumptions on masses, elastic constants and positions of the walls: 
\begin{itemize}
\item[$\mathcal{A}_1)$] The set $\S^{(n)}:=(s_1^{(n)},\dots,s_n^{(n)})$ of the equilibrium positions is contained in a bounded 
interval around the origin, $s_j^{(n)} \in [-L/2,L/2] \;\; 
\forall n$ and $\forall j=1\dots n$;
\item[$\mathcal{A}_2)$] The two measures $\mu_M^{(n)}:=\sum_{j=1}^{n} M_j^{(n)} \delta_{s_j^{(n)}}$ and 
$\mu_K^{(n)}:=\sum_{j=1}^n K_j^{(n)} \delta_{s_j^{(n)}}$, 
where $\delta_{s_j^{(n)}}$ is the Dirac mass supported at point
$s_j^{(n)}$,  are uniformly equivalent in the sense that
$0<c_1<\sup\big[L^2 K_j^{(n)}/ a^2 M_j^{(n)}\big]<c_2<\infty$,
 here $c_1$ and $c_2$ are two positive constants. The two measures have total mass which is bounded uniformly in $n$ :  $\sum_{j=1}^{n} M_j^{(n)}<c\;\;\;$ $\sum_{j=1}^{n} K_j^{(n)}<c$. Furthermore $\mu_M^{(n)}$ and $\mu_K^{(n)}$ 
weakly converge, when $n$ tends to infinity, to measures supported in $[-L/2,L/2]$  which are absolutely continuous with respect to Lebesgue measure (the corresponding Radon-Nikodim derivatives will be denoted with 
$S \rho_M$ and $S \rho_K$);
\end{itemize}

\n
The initial conditions will be chosen independent of $n$. In particular we consider the case
\be\label{inco2}
p^{(n)}(x,0)=p_0(x);\;v^{(n)}(x,0)=v_0(x);\; y_j^{(n)}(0)=0;\;z_j^{(n)}(0)=0\qquad j=1,\dots,n\,,
\ee
\n where $p_0$ and $v_0$ are in $H^2(\RE)$ and $\supp(p_0)\cap[-L/2,L/2]=\emptyset$, $\supp(v_0)\cap[-L/2,L/2]=\emptyset$.
\begin{remark}
The request of equivalence of the two measures $\mu_M^{(n)}$ and 
$\mu_K^{(n)}$ could be released in various ways but it is well motivated, from a physical 
point of view, preventing the occurrence of zero or infinite proper oscillation frequencies. 
\end{remark}

\n Initial conditions more general than  (\ref{inco2}) might of course be considered. Nevertheless 
(\ref{inco2}) seem to be  the most natural initial  conditions which do not depend on $n$. Notice, in particular, 
that at each step $n$, corresponding to the zero eigenvalue of $\hat A^{(n)}$, there are static solutions 
of problem (1)-(4). As it was shown in \cite{CFP06} the static solutions are of the form  $\Psi_{st}=(p_{st}\nn,0,\uy_{st},\underline0)$ with $\supp(p_{st})\subseteq[-L/2,L/2]$.  Being initial conditions  (\ref{inco2}) orthogonal to all 
static solutions we are guaranteed that a genuine scattering problem is analyzed.

\n Under the assumptions stated above we prove the following
\begin{theorem} 
Let the sequence of sets $\S^{(n)}:=(s_1^{(n)},\dots,s_n^{(n)})$ and the sets of positive real numbers 
$\underline{M}^{(n)}:=(M_1^{(n)},\dots,M_n^{(n)})$ and $\underline{K}^{(n)}:=(K_1^{(n)},\dots,K_n^{(n)})$ 
be such that the assumptions $\mathcal{A}_1)$ and $\mathcal{A}_2)$ above are satisfied.
Let $v^{(n)}\in C^2(\RE;L^2(\RE))\cap C^1(\RE;H^1(\RE))\cap 
C(\RE;H^2(\RE\backslash\S\nn)) $ 
be  the velocity field in the solution of problem \eqref{cpn1} - \eqref{cpn2} with initial conditions 
satisfying assumption $\mathcal{A}_3)$.
Then  
$$
\text{$v^{(n)}$ and $\frac{\partial v^{(n)}}{\partial t}$ converge to $v$ and $\frac{\partial v}
{\partial t}$ uniformly over compact subsets of $\RE^2$}
$$
and
$$
\text{$\frac{
  \partial v^{(n)}}{\partial x}(t,\cdot )$ converges to $\frac{ \partial v}{\partial x}(t,\cdot )$ in $L_{loc}^2(\RE)$,
uniformly in $t$ over compact intervals,}
$$
where $v\in C^{1}(\RE;C(\RE))\cap C(\RE;H^{1}_{loc}(\RE))$ 
is the unique strong solution\footnote{see equation \eqref{strong} below for more precise regularity properties of $v$.} of the hyperbolic equation
\be 
\label{hypeq}
\Big(1+\frac{\rho_M}{\rho_0}\Big)\, \pd{^2v}{t^2}-
a^2\, \pd{^2v}{x^2}
+\frac{\rho_K}{\rho_0}\,v=0.
\ee
with initial data
\be\label{hypeq2}
v|_{t=0}=v_0\,;\qquad
\pd{v}{t}\bigg|_{t=0}=-\frac{1}{\rho_0}\frac{dp_0}{dx}\,.
\ee
\end{theorem}
\begin{proof}
By the basic properties of one-parameter 
groups of operators and their
generators (see e.g. \cite{G}), 
the map $t\mapsto e^{t\hat A^{(n)}}\Psi_0$ belongs to 
$C^1(\RE;\HH)\cap C(\RE;D(\hat A^{(n)}))$ if and only 
if $\Psi_0\in D(\hat A^{(n)})$. In this case 
\[
\frac{d\,}{dt}\,e^{t\hat A^{(n)}}\Psi_0=\hat A^{(n)}e^{t\hat A^{(n)}}
\Psi_0=
e^{t\hat A^{(n)}}\hat A^{(n)}\Psi_0\,,
\]
and, since $\hat A^{(n)}$ is skew-adjoint,
\be\label{sa1}
\|\hat A^{(n)}
\Psi_t\nn\|=\|e^{t\hat A^{(n)}}\hat A^{(n)}\Psi_0\|=
\|\hat A^{(n)}
\Psi_0\|\,.
\ee
Therefore, since any initial condition
$\Psi_0=(p_0,v_0,\underline0,\underline0)$ chosen according to
assumption $\mathcal{A}_3)$ is in $D\big([\hat A^{(n)}]^2\big)$ 
for any $n$, the path $t\mapsto 
\Psi_t\nn:=e^{t\hat A\nn }\Psi_0$ is in 
$C^2(\RE;\HH)\cap C^1(\RE;D(\hat A^{(n)}))\cap 
C(\RE;D([\hat A^{(n)}]^2))$, 
\be\label{d2}
\frac{d^2\Psi_t\nn}{dt^2}=[\hat A^{(n)}]^2
\Psi_t\nn
\ee
and
\be\label{sa2}
\|[\hat A^{(n)}]^2
\Psi_t\nn\|=\|e^{t\hat A^{(n)}}[\hat A^{(n)}]^2\Psi_0\|=
\|[\hat A^{(n)}]^2
\Psi_0\|\,.
\ee
\n From the explicit characterization of domain and action of $[A\nn]^2$
\[ 
\begin{aligned}
&D([\hat A\nn]^2)
=\bigg\{(p\nn,v\nn,\uy\nn,\uz\nn)\in D(A\nn)\bigg|\, 
\frac{dv\nn}{dx}\in H^1(\RE\backslash \S\nn),\\
 &\frac{dp_{reg}\nn}{dx}\in
H^1(\RE),\,
\frac{dp_{reg}\nn}{dx}(s_j\nn)
=\rho_0\bigg(\frac{K_j\nn}{M_j\nn}y_j\nn+\frac{S}{M_j\nn}\sigma_j\nn\bigg),\,j=1,\dots,n\bigg\}\,,
\end{aligned}
\]
\[
\begin{aligned}
&[\hat A\nn]^2(p\nn,v\nn,\uy\nn,\uz\nn)=\Bigg(a^2\frac{d^2p_{reg}\nn}{dx^2},\,
a^2\frac{d}{dx}\bigg(\frac{dv\nn}{dx}\bigg)_{reg},\\
&-\sum_{j=1}^n\bigg(\frac{K_j\nn}{M_j\nn}\,y_j\nn
+\frac{S}{M_j\nn}\,\sigma_j\nn\bigg)\,\ue_j\nn,\,
-\sum_{j=1}^{n}\bigg(\frac{K_j\nn}{M_j\nn}\,z_j\nn-\frac{a^2\rho_0S}{M_j\nn}\,\zeta_j\nn\bigg)\,\ue_j\nn\Bigg)
\end{aligned}
\]
where
\[
\zeta_j\nn:= \frac{dv\nn}{dx}(s_j^{(n)+})-\frac{dv\nn}{dx}(s_j^{(n)-}),
\]
\[
\sigma_j\nn:=p\nn(s_j^{(n)+})-p\nn(s_j^{(n)-}),
\]
\be\label{dervreg}
\bigg(\frac{dv\nn}{dx}\bigg)_{reg}(x):=\frac{dv\nn}{dx}(x)-\frac{1}{2}\sum_{j=1}^n\zeta_j\nn\sgn(x-s_j\nn)\,.
\ee
\n  We conclude that if $\Psi\nn_t=(p\nn(t),v\nn(t),
\uy\nn(t),\uz\nn(t))$ is the solution of \eqref{d2}, then
$$
v\nn\in C^2(\RE;L^2(\RE))\cap C^1(\RE;H^1(\RE))\cap 
C(\RE;H^2(\RE\backslash\S\nn))\,,\quad \uz\in C^2(\RE)\,,
$$
and $(v\nn,\uz\nn)$ is the strong solution of the following 
Cauchy problem with time-dependent boundary conditions:
\begin{align}
&\frac{\partial^2v\nn}{\partial t^2}=a^2\frac{\partial}{\partial x}
\bigg(\frac{\partial v\nn}{\partial x}
\bigg)_{reg}
\,,\label{tilde}\\
&\frac{d^2z_j\nn}{dt^2}+\frac{K_j\nn}{M_j\nn}\,z_j\nn=
\frac{a^2\rho_0S}{M_j\nn}\,\zeta_j\nn\,,\quad j=1,\dots,n\,,\\
&v\nn(t,s_j\nn)=z_j\nn(t)\,,\quad \frac{\partial v\nn}{\partial t}(t,s_j\nn)=
\frac{dz_j\nn}{dt}(t)\label{bc}
\,,\quad j=1,\dots,n\,,\\
&v\nn\big|_{t=0}=v_0\,,\quad \frac{\partial v\nn}{\partial t}\bigg|_{t=0}
=-\frac{1}{\rho_0}\frac{\partial p_0}{\partial x}\,,\quad \\
&\uz\nn\big|_{t=0}=0\,,\quad
\frac{d\uz\nn}{dt}\bigg|_{t=0}=0\,.
\end{align}
\n Furthermore, by energy conservation and by \eqref{sa1} and \eqref{sa2}, there exists  a positive constant $c$
depending only on the norm of the initial datum $\Psi_0=(p_0,v_0,\underline0,\underline0)$, such that 
the following bounds hold for any $t$ and $n$:
\be
\left\|v\nn(t,\cdot)\right\|_{H^1(\RE)}\le  c\,,\quad
\left\|\frac{\partial v\nn}{\partial t} (t,\cdot )\right\|_{H^1(\RE)} \le c\,,\quad
\left\|\frac{\partial^2v\nn}{\partial t^2}(t,\cdot) \right\|_{L^2(\RE)}\le 
c\,,\label{derti} 
\ee
\be\label{reg}
\left\|\bigg(\frac{\partial v\nn}{\partial x}\bigg)_{reg}(t,\cdot)\right\|_{H^1(\RE)}\le  c\,,
\ee
\be
\sum_{j=1}^{n}M_j\nn|z_j\nn(t)|^2\le c\,,\quad
\sum_{j=1}^{n}M_j\nn\bigg|\frac{K_j\nn}{M_j\nn}\,z_j\nn(t)-\frac{a^2\rho_0S}{M_j\nn}\,\zeta_j\nn(t)\bigg|^2\le c\,.
\label{bound4}
\ee
By the continuous injection $H^1(\RE)\hookrightarrow C_b(\RE)$ one has 
$$
\sup_{(t,x)\in\RE^2}|v\nn(t,x)|\le c_1\sup_{t\in\RE}\left\|v\nn(t,\cdot)\right\|_{H^1(\RE)}
$$
and (here $x_1<x_2$ and $t_1<t_2$)
\begin{align*}
|v\nn(t_1,x_1)-v\nn(t_2,x_2)|
\le &\sqrt{x_2-x_1}\,\sup_{t\in\RE}\left\|\frac{\partial v\nn}{\partial x}(t,\cdot)\right\|_{L^2(\RE)}+(t_2-t_1)\,\sup_{(t,x)\in\RE}\left|\frac{\partial v\nn}{\partial t}(t,x)\right|\\
\le &\sqrt{x_2-x_1}\,\sup_{t\in\RE}\left\|v\nn (t,\cdot)\right\|_{H^1(\RE)}+c_1
(t_2-t_1)\,\sup_{t\in\RE}\left\|\frac{\partial v\nn}{\partial t}(t,\cdot)\right\|_{H^1(\RE)}\,.
\end{align*}
Hence by \eqref{derti} the sequence $\{v\nn\}_1^\infty$ is equi-bounded and equi-continuous. Thus, by Ascoli-Arzela theorem, there exist a subsequence $\{v^{(n_{k})}\}_1^\infty$ and $v\in C(\RE^2)$ such that $v^{(n_{k})}$ converges to $v$ uniformly over compact subsets of $\RE^2$.
\par Let us now take a bounded rectangle $R=I_1\times I_2\subset\RE^2$. 
Since the injection
$H^1(I_2)\hookrightarrow C(\bar I_2)$ is a compact operator and, by \eqref{derti}, 
$\left\{\frac{\partial v^{(n_{k})}}{\partial t}\right\}_1^\infty$ is bounded in 
$L^\infty(I_1; H^1(I_2))$ and $\left\{\frac{\partial^2 v^{(n_{k})}}{\partial t^2}\right\}_1^\infty$ is bounded in $L^\infty(I_1;L^2(I_2))$, by corollary 4  in \cite{JS},  
section 8, applied to the triple of Banach spaces $X=H^1(I_2)\subset B=C(\bar I_2)\subset Y=L^2(I_2)$,  
there exists a sub-subsequence $ \left\{\frac{\partial v^{(n^R_{i})}}{\partial t}\right\}_1^\infty$ converging to $\dot v_R$ in $C(\bar I_1;C(\bar I_2))$. 
\par Let $H^\alpha(I_2)$, $\alpha\in(0,1)$, be defined, as usual, by
$$
H^\alpha(I_2):=\left\{f\in L^2(I_2)\,:\, \int_{I_2}\int_{I_2}\frac{|f(x)-f(y)|^2}{|x-y|^{1+2\alpha}}\,dxdy<\infty\right\}\,.
$$
$H^\alpha(I_2)$ is a Banach space with norm 
$$
\|f\|_{H^\alpha(I_2)}=\left(\|f\|^2_{L^2(I_2)}+\int_{I_2}\int_{I_2}\frac{|f(x)-f(y)|^2}{|x-y|^{1+2\alpha}}\,dxdy\right)^{1/2}\,.
$$
Notice that $\sgn(\cdot-s_j\nn)\in H^\alpha(I_2)$ for any $\alpha<1/2$ and for any bounded $I_2\subset\RE$ and its $H^\alpha(I_2)$ norm is bounded uniformly in $n$. Hence, by \eqref{dervreg}, \eqref{reg}, \eqref{bound4} and $\mathcal A_2$),  $\left\{\frac{\partial v^{(n_{k})}}{\partial x}\right\}_1^\infty$ is bounded in 
$L^\infty(I_1; H^\alpha(I_2))$, $\alpha<1/2$. Moreover, by \eqref{derti}, $\left\{\frac{\partial}{\partial t}\frac{\partial v^{(n_{k})}}{\partial x}\right\}_1^\infty$ is bounded in $L^\infty(I_1;L^2(I_2))$. Thus, since the injection
$H^\alpha (I_2)\hookrightarrow L^2( I_2)$, $\alpha>0$, is a compact operator, 
again by corollary 4  in \cite{JS},  
section 8, applied to the triple of Banach spaces $X=H^\alpha(I_2)\subset B=Y=L^2(I_2)$,  
there exists a sub-subsequence $ \left\{\frac{\partial v^{(n^R_{j})}}{\partial x}\right\}_1^\infty$ converging to $v'_{R}$ in $C(\bar I_1;L^2(I_2))$. 
\par
Considering then the sequence of rectangles $R_m=(-m,m)\times(-m,m)$, by a standard 
diagonal argument, one obtains a further subsequence, which for notational convenience we  continue to denote by $\{v^{(n_k)}\}_1^\infty$, such that $ \left\{\frac{\partial v^{(n_{k})}}{\partial t}\right\}_1^\infty$ converges to $\dot v\in C(\RE^2)$ uniformly over compact sets, and such that $ \left\{\frac{\partial v^{(n_{k})}}{\partial x}\right\}_1^\infty$ 
converges to $v'\in C(\RE;L_{loc}^2(\RE))$ uniformly over compact time intervals. 
Defining $$v_1(t,x):=v(0,x)+\sgn(t)\int_0^t\dot v(s,x)\,ds\,,\quad v_2(t,x):=v(t,0)+\sgn(x)\int_0^x v'(t,y)\,dy\,,
$$ one obtains $v^{(n_k)}\to v_1$ and  $v^{(n_k)}\to v_2$, and so $
v_1=v_2=v$, $\dot v=\frac{\partial v}{\partial t}$, $v'=\frac{\partial v}{\partial x}$.
\par
Let us now multiply both sides of equation 
\eqref{tilde} by a function 
$g\in C_c^{1}(\RE^2)$. 
Integrating by part we get
\[
\begin{aligned}
&-\int_0^\infty\left\langle\frac{\partial  v^{(n_k)}}{\partial t}(t,\cdot), \frac{\partial  g}{\partial t}(t,\cdot)\right\rangle\,dt
-\left\langle\frac{\partial v^{(n_k)}}{\partial t}(0,\cdot),g(0,\cdot)
\right\rangle
=
-a^2\int_0^\infty\left\langle\frac{\partial v^{(n_k)}}{\partial x}(t,\cdot),\frac{\partial g}{\partial x}(t,\cdot)\right\rangle\,dt\\ 
+&\frac{1}{\rho_0S}\int_0^\infty \sum_{j=1}^{n_k}M_j^{(n_k)}\,
\frac{dz_j^{(n_k)}}{dt}\,
(t)\,\frac{\partial g}{\partial t}(t,s_j^{(n_k)})\,dt-\frac{1}{\rho_0S}\int_0^\infty \sum_{j=1}^{n_k}K_j^{(n_k)} z_j^{(n_k)}(t) g(t,{s_j^{(n_k)}})\,dt\,.
\end{aligned}
\]
Using assumption $\mathcal{A}_2)$ and $\mathcal{A}_3)$  we conclude that
\[
\begin{aligned}
&\lim_{k\to\infty}\frac{1}{\rho_0S}\int_0^\infty \sum_{j=1}^nM_j^{(n_k)}\,
\frac{dz_j^{(n_k)}}{dt}\,
(t)\frac{\partial g}{\partial t}(t,s_j^{(n_k)})\,dt
\\
=&\lim_{k\to\infty}\frac{1}{\rho_0S}\int_0^\infty \sum_{j=1}^nM_j^{(n_k)}\,
\left(\frac{\partial v^{(n_k)}}{\partial t}\,(t,s_j^{(n_k)})
-\frac{\partial v}{\partial t}\,(t,s_j^{(n_k)})\right)\,\frac{\partial g}{\partial t}(t,s_j^{(n_k)})\,dt\\
&+\lim_{k\to\infty}\frac{1}{\rho_0S}\int_0^\infty \sum_{j=1}^nM_j^{(n_k)}\,
\frac{\partial v}{\partial t}(t,s_j^{(n_k)})\,\frac{\partial g}{\partial t}(t,s_j^{(n_k)})\,dt\\
=&\int_0^\infty\left\langle
\frac{\rho_M}{\rho_0}\,\frac{\partial v}{\partial t}(t,\cdot),\frac{\partial g}{\partial t}(t,\cdot)\right\rangle\,dt
\end{aligned}
\]
and
\[
\begin{aligned}
&\lim_{k\to\infty}
\frac{1}{\rho_0S}\int_0^\infty \sum_{j=1}^{n_k}K_j^{(n_k)} z_j^{(n_k)}(t) g(t,{s_j^{(n_k)}})\,dt
\\
=&\lim_{k\to\infty}\frac{1}{\rho_0S}\int_0^\infty \sum_{j=1}^{(n_k)}\,
(v^{(n_k)}(t,s_j^{(n_k)})-v(t,s_j^{(n_k)}))\,g(t,s_j^{(n_k)})\,dt\\
&+\lim_{k\to\infty}\frac{1}{\rho_0S}\int_0^\infty \sum_{j=1}^nK_j^{(n_k)}\,
v(t,s_j^{(n_k)})\,g(t,s_j^{(n_k)})\,dt\\
=&\int_0^\infty\left\langle
\frac{\rho_K}{\rho_0}\,v(t,\cdot),g(t,\cdot)\right\rangle\,dt\,.
\end{aligned}
\]
The limit function $v$ then satisfies the equation (notice that $\frac{\partial v}{\partial t}(0,\cdot)$ and $\rho_M$ have disjoint supports)
\[ 
\begin{aligned}
&- \int_0^\infty\left\langle\Big(1+\frac{\rho_M}{\rho_0}\Big)\frac{\partial v}{\partial t}(t,\cdot),\frac{\partial g}{\partial t}(t,\cdot)\right\rangle\,dt -\left\langle\Big(1+\frac{\rho_M}{\rho_0}\Big)\frac{\partial v}{\partial t}(0,\cdot),\frac{\partial g}{\partial t}(0,\cdot)\right\rangle\\
 =&-a^2\int_0^\infty\left\langle\frac{\partial v}{\partial x}(t,\cdot),\frac{\partial g}{\partial x}(t,\cdot)\right\rangle\,dt
-\int_0^\infty\left\langle
\frac{\rho_K}{\rho_0}\,v(t,\cdot),g(t,\cdot)\right\rangle\,.
\end{aligned}
\]
which is a weak form of equation \eqref{hypeq} when the initial data $v_0$ is specified.
\par
Posing $w:=1+\frac{\rho_M}{\rho_0}$ and $q:=\frac{\rho_K}{\rho_0}$,  let us define the weighted $L^2$ space
$$
L^2(\RE;w):=\{f:\RE\to\CO\ \text{measurable}\,:\,f\sqrt w\in L^2(\RE) \}
$$
and the linear operator
$$
L:D(L)\subseteq L^2(\RE;w)\to L^2(\RE;w)\,,\qquad  Lf:=\frac{1}{w}\,
\left(-a^2\,\frac{d^2f}{dx^2}+qf\right)\,,
$$
$$
D(L):=\left\{f\in L^2(\RE;w)\,:\, f\in C^1(\RE)\,,\ \frac{df}{dx}\in \text{AC}(\RE)\,,\quad Lf\in L^2(\RE;w)\right\}\,.
$$
Here $\text{AC}(\RE) $ denotes the set of absolutely continuous functions on the real line. 
Notice that $L$ is in the limit point case at both  $+\infty$ and $-\infty$ (use e.g. Theorem 6.3 in \cite{W}) and so, by  the theory of Sturm-Liouville operators (see e.g. \cite {W}), $L$ is 
self-adjoint (and positive). Hence, by the theory of abstract second-order equations (see e.g. \cite{G}, Chapter 2, Section 7), the Cauchy problem 
\begin{align*}
&w\frac{\partial^2v }{\partial t^2}=a^2\frac{\partial^2 v}{\partial x^2}-qv\\
&v|_{t=0}=v_1\in D(L)\,,\quad \frac{\partial v}{\partial t}\bigg|_{t=0}=v_2\in D(\sqrt L)
\end{align*}
has an unique strong solution 
\be\label{strong} 
v\in C^2(\RE;L^2(\RE;w))\cap C^1(\RE;D(\sqrt L))\cap 
C(\RE;D(L))\,.
\ee
Since, by hypotheses ${\mathcal A}_2$) and ${\mathcal A}_3$), $v_0\in D(L)$ and $\frac{dp_0}{dx}\in D(\sqrt L)$, the regularity of the limit function $v$ can be precised further: 
it is the unique strong solution of Cauchy problem \eqref{hypeq}-\eqref {hypeq2} and satisfies  \eqref{strong}. 
\par 
Suppose now that the whole sequence $\{v^{(n)}\}_{1}^{\infty}$ does not converges to $v$. Then there exists a subsequence $\{v^{(n'_k)}\}_{1}^{\infty}$ such that for all $k$ and some bounded rectangle $R'\subset\RE^{2}$,
$$
\sup_{(t,x)\in R'}|v^{(n'_k)}(t,x)-v(t,x)|\ge \epsilon>0
$$ 
(analogous relations hold for $\frac{\partial v^{(n'_{k})}}{\partial t}$ and $\frac{\partial v^{(n'_{k})}}{\partial x}$). Since $\S^{(n'_{k})}$, $\underline{M}^{(n'_{k})}$ and $\underline{K}^{(n'_{k})}$ satisfy the assumptions $\mathcal{A}_1)$ and $\mathcal{A}_2)$ above, we get a contradiction and so the whole sequence converges to $v$.
\end{proof}
\vspace{.7cm}
\n
We want to conclude the paper with few short remarks on the effective equation (\ref{hypeq}) we obtain in the limit. 

\n
It is worth stressing that (\ref{hypeq}) is not simply a wave equation with space-dependent propagation speed $\displaystyle \frac{a \sqrt{\rho_0}}{\sqrt{\rho_0 + \rho_M}}\,$ taking into account the masses of the oscillators.

\n
In fact, the potential term $\displaystyle \frac{\rho_{K}}{a^2 \rho_0}$ describes a source, depending linearly on the velocity field, due to the interaction of the field with the oscillators. In this respect, the problem we addressed is not a classical one in homogenization of the wave equation  in highly inhomogeneous or random media. On the other hand, non-linear coupling between field and sources can be easily modeled. 

\n
The higher-dimensional analogue of the problem treated in this paper is under investigation. Using techniques of self-adjoint extensions of symmetric operators, as we did in \cite{CFP06}, we intend to make available two and three dimensional models of finite or infinite point emitters interacting via their own-generated acoustic field.

\end{document}